\documentclass{article}
\usepackage{amsmath,amssymb,amsthm}
\usepackage{graphicx}

\newtheorem{theorem}{Theorem}
\newtheorem{lemma}{Lemma}
\theoremstyle{definition}
\newtheorem{definition}{Definition}

\author{Svein Høgemo}
\title{Linear MIM-width of the Square of Trees}

\begin{document}

\maketitle

\begin{abstract}
Graph parameters measure the amount of structure (or lack thereof) in a graph that makes it amenable to being decomposed in a way that facilitates dynamic programming. Graph decompositions and their associated parameters are important both in practice (as a tool for designing robust algorithms for NP-hard problems) and in theory (relating large classes of problems to the graphs on which they are solvable in polynomial time).

Linear MIM-width is a variant of the graph parameter MIM-width, introduced by Vatshelle. MIM-width is a parameter that is constant for many classes of graphs. Most graph classes which have been shown to have constant MIM-width also have constant linear MIM-width. However, computing the (linear) MIM-width of graphs, or showing that it is hard, has proven to be a huge challenge. To date, the only graph class with unbounded linear MIM-width, whose linear MIM-width can be computed in polynomial time, is the trees. In this follow-up, we show that for any tree $T$ with linear MIM-width $k$, the linear MIM-width of its square $T^2$ always lies between $k$ and $2k$, and that these bounds are tight for all $k$.
\end{abstract}

\section{Introduction}

Divide-and-conquer is a tried and true strategy for obtaining fast algorithms. In addition to its classical use for polynomial-time algorithms, it is also a valuable tool for designing parameterized algorithms for NP-hard problems. For problems on graphs, this strategy often manifests in the form of \emph{graph decompositions}. 

A graph decomposition is a tree structure that relates to the graph in such a way that one can solve any of a number of problems on the graph by executing dynamical programming on the decomposition. However, not all graphs will be equally suitable for a given decomposition (nor should we expect them to be, as we assume NP-hard problems are hard to solve on some instances). The \emph{width} of a decomposition is a number $k$ that measures how badly a graph fits the decomposition -- the lower the $k$, the better the structure of the graph fits the particular decomposition technique. The goal is to solve the problem in time polynomial in the size of the graph, but exponential in $k$. The arguably most famous decomposition technique, \emph{tree decomposition}, and its associated parameter \emph{treewidth} works well on many classes of sparse graphs -- trees have treewidth 1, and in general graphs with low treewidth have a ``tree-like'' structure. Due to the large role treewidth has played in structural graph theory, and its usefulness in designing algorithms for a plethora of problems~\cite{de2006exploiting,madani2017finding,maniu2019experimental,ordyniak2013parameterized}, the problem of computing tree decompositions of decent treewidth has received great attention, and several algorithms, both exact and approximating, have been proposed, with varying usability in practice~\cite{bodlaender1993linear,bodlaender2010treewidth,bodlaender2011treewidth,korhonen2022single}. The different uses of treewidth are well explained in Bodlaender's series of surveys on the matter~\cite{bodlaender1997treewidth,bodlaender2006treewidth,bodlaender2007treewidth,koster2001treewidth}.

Since tree decomposition only works well for sparse graphs, several other decomposition techniques have been proposed. One quite versatile technique is \emph{branch decomposition}, where the vertices of the graph are mapped to the leaves of a tree of maximum degree 3. Each edge in the tree corresponds to a cut of the graph (defined by which side of the edge the leaf corresponding to each vertex lies on), and there have been defined several width measures on these decompositions, based on functions on the cuts. One such width measure that has garnered significant theoretical interest is \emph{MIM-width}. ``MIM'' stands for ``maximum induced matching'', and MIM-width measures the the biggest induced matching in any of the bipartite graphs induced by the cuts defined by the decomposition. MIM-width was introduced by Vatshelle~\cite{vatshelle2012new} where its strongness and algorithmic properties were expounded.

The significance of MIM-width lies in how veritably strong it is. Several important graph classes have constant MIM-width, among them interval graphs and circular arc graphs, permutation graphs, complements of graphs with constant degeneracy, graphs with constant treewidth~\cite{vatshelle2012new}, and powers of graphs with constant MIM-width~\cite{jaffke2019mim}. The trade-off is a worse dependence on the parameter than in most other types of decomposition -- typical running times of algorithms parameterized by MIM-width are $n^{O(k)}$ or $n^{O(k^2)}$ for graphs with MIM-width $k$, given a decomposition (see e.g. the algorithms given in \cite{bergougnoux2021more,jaffke2018unified,jaffke2020mim}). In other words, for any constant $k$ a polynomial time algorithm exists, but the larger $k$ is, the higher the degree of the polynomial is (this is called an $XP$ algorithm). Still, obtaining an algorithm parameterized by MIM-width means that the problem is solvable in polynomial time for many graph classes.

For all width measures of branch decompositions, one can define a ``linear'' variant, where the allowed decompositions are restricted to linear layouts of the graph. Linear graph parameters are interesting, both as a sort of test-bed for proving things about the decompositions (as they are easier to reason about), but also because algorithms that work on linear layouts are faster in practice than algorithms that work on trees. Furthermore, MIM-width and linear MIM-width have the peculiar relationship that most graph classes that have constant MIM-width also have constant linear MIM-width. In fact, of the graph classes that have unbounded treewidth or clique-width, only the class of leaf powers and some of its subclasses, such as the trees themselves, are known to have bounded MIM-width and unbounded linear MIM-width~\cite{jaffke2020bounded}.

Without doubt, the greatest challenge regarding (linear) MIM-width is the problem of computing it on arbitrary graphs. Several hardness results exist: In \cite{saether2016hardness}, it was proved that it is at least as hard to compute the (linear) MIM-width of a graph, or an optimal branch decomposition, as it is to compute the MIM of a graph. Maximal induced matching is itself a hard problem in several ways; in addition to being NP-complete, it is hard to approximate to a constant factor in polynomial time~\cite{elbassioni2009approximability}, and also it is hard for the parameterized complexity class $W[1]$, meaning that we likely will never find any algorithm that decides whether an arbitrary graph has a maximal induced matching of size at least $k$ in time bounded by $f(k)\cdot n^{O(1)}$~\cite{moser2009parameterized}. All of these barriers carry over to the problem of computing (linear) MIM-width. This is coupled with a lack of positive results: To date, no algorithm has been found that in polynomial time decides whether an arbitrary graph has (linear) MIM-width at most $k$, for any constant $k$. Even for (linear) MIM-width 1, no algorithm has been found that in polynomial time outputs a decomposition of \emph{any} constant width, or concludes that the graph has width $>1$. It is not implausible that recognizing graphs with (linear) MIM-width $k$ is NP-complete for some (small) $k$, especially in light of a recent, similar result regarding twin-width, another strong parameter~\cite{berge2022deciding}. If this turns out to be true, it would naturally be a huge disadvantage, although not taking away from its proven usefulness for the many graph classes with bounded MIM-width.

Regarding positive results, for all graph classes with proven bounded (linear) MIM-width, there is an easy way of finding a good layout or decomposition; for example, an optimal layout of an interval graph is found by ordering the vertices in order of the left endpoints of their respective intervals~\cite{belmonte2013graph,vatshelle2012new}. Beware that for graph classes with a higher bound than 1, the decompositions may not be optimal, they are only guaranteed to be bounded. To date, the only polynomial-time algorithm for computing (linear) MIM-width on a class of unbounded width is the one given in \cite{hogemo2019linear} for the linear MIM-width of trees (the MIM-width of trees is 1). Om the class of trees, the linear MIM-width is at most logarithmic and within a constant factor of their pathwidth (the linear variant of treewidth).

This discrepancy between the use of MIM-width in designing robust graph algorithms and how little we know about its computability, makes it hard to gauge the exact potential of this parameter. Studying the linear MIM-width of simple graph classes, like the squares of trees, can help illuminating what makes the problem difficult to solve in the general case, and find out when the structural tools we have for analysing trees break down. Indeed, if deciding (linear) MIM-width $\leq k$ for some constant $k$ actually proves NP-complete, it will likely be proved by restricting the problem to some graph class on which the problem is also NP-complete. This is not unusual, as restricting the problem to a special case makes reductions easier to find (see e.g. the classic result for the the NP-completeness of chordal completion~\cite{yannakakis1981computing}).

Following the result in \cite{jaffke2019mim}, any power of a graph $G^k$ has at most twice the (linear) MIM-width of the original graph $G$. Therefore powers of trees have at most logarithmic linear MIM-width. This is in stark contrast to e.g. pathwidth, since the square of a star is a complete graph and therefore has linear pathwidth. Furthermore, for any graph $G$, $G^{diam(G)}$ is a complete graph and therefore has linear MIM-width 1. Squares of trees are a simple graph class that nevertheless have more going on than trees themselves, and therefore serves as a natural class to extend the research from \cite{hogemo2019linear} in.

The rest of the paper is organized as follows: Section 2 introduces the notation necessary to follow the paper; Section 3 contains the main result and the previous results that we use in order to prove the main result; and finally, Section 4 concludes with a short discussion on its implications (if any).

\section{Preliminaries}

All graphs considered are finite and simple. $V(G)$ and $E(G)$ denotes the sets of vertices and edges in the graph $G$, respectively.

$N_G(v)$ denotes the \emph{open neighborhood} of the vertex $v$ in the graph $G$, i.e. the set of vertices with which $v$ shares an edge. $N_G[v]$ denotes the \emph{closed neighborhood}, i.e. $N(v)\cup\{v\}$. For a set of vertices $S\subseteq V(G)$, $N[S] = \bigcup_{v\in S}N[v]$, and $N(S) = N[S]\setminus S$. If $G$ is obvious from context, subscripts can be omitted.

Let $S,T$ be two disjoint subsets of $V(G)$. $G[S,T]$ is the bipartite graph induced by $S$ and $T$, i.e. the subgraph consisting of all edges with one endpoint in $S$ and the other in $T$.

If $T$ is a rooted tree, the notation $T[v]$ for some node $v$ refers to the subtree rooted in $v$, i.e. the rooted tree consisting of $v$ (as root) and all its descendants.

The \emph{distance} between two vertices $u,v\in V(G)$ is the length of a shortest path between $u$ and $v$. Given two subgraphs $H,H'\subseteq G$, the distance between $H$ and $H'$ is the minimum distance between any vertex in $H$ and any vertex in $H'$. For example, the distance between $H$ and $H'$ is at least 1 iff $H$ and $H'$ are disjoint. The \emph{diameter} of $G$, $diam(G)$, is the length of the longest distance between any two vertices in $G$.

\begin{definition}[graph power]
Given a graph $G$ and some integer $k\geq 1$, the $k$-th \emph{power} of $G$, denoted $G^k$, is
the unique graph that has $V(G^k) = V(G)$,
and the additional property that for any two vertices $u,v$, $u$ and $v$ are
adjacent in $G^k$ if and only if they have distance at most $k$ in $G$. If we
consider every vertex in $G$ ``adjacent to'' itself, we have that $G^1 = G$, and furthermore, given the adjacency matrix of $G$ as a boolean
matrix $\mathcal{A}$, the adjacency matrix of $G^k$ is given by $\mathcal{A}^k$.

We define the \emph{square} of $G$ as $G^2$, the second power of $G$.
\end{definition}

\begin{definition}[linear layout]
A \emph{linear layout} of a graph $G$ is a bijection $\sigma:V(G)\rightarrow \{1,\ldots,|V(G)|\}$, i.e. a total order on the vertices of $G$. We will routinely use $v_i$ as a shorthand for $\sigma^{-1}(i)$ when $G$ and $\sigma$ is given.
\end{definition}

The next three definitions are found in \cite{hogemo2019linear}:

\begin{definition}[maximum induced matching, mim-width of a layout]
For a graph $G$ on $n$ vertices, we denote by $\mathrm{mim}(G)$ the size of its \emph{maximum induced matching} (MIM), the largest number of edges whose endpoints induce a matching. Let $\sigma$ be a linear layout of $G$. For any index $1\leq i< n$ we have a subset of $V(G)$, $V^\sigma_i = \{v_1,\ldots,v_i\}$. We call the partition $(V^\sigma_i,\overline{V^\sigma_i})$ a \emph{cut} of $G$. The \emph{maximum induced
matching width}, or \emph{MIM-width} of G under layout $\sigma$ is denoted $\mathrm{mw}(\sigma, G)$, and is defined as the maximum, over all $1\leq i< n$, of $mim(G[V^\sigma_i,\overline{V^\sigma_i}])$.
\end{definition}

\begin{definition}[linear MIM-width]
The \emph{linear maximum induced matching width} -- \emph{linear MIM-width} -- of $G$ is denoted $\mathrm{lmw}(G)$, and is the minimum value of $\mathrm{mw}(\sigma,G)$ over any possible layout $\sigma$ of the vertices of $G$.
\end{definition}

Note: When talking about a cut $(V^\sigma_i,\overline{V^\sigma_i})$ in a graph $G$, a vertex $v$ is said to lie \emph{to the left of} the cut if and only if $v \in V^\sigma_i$; otherwise,
$v$ is said to lie \emph{to the right of} the cut.

\begin{definition}[$k$-neighbor]
Let $x$ be a node in the tree $T$ and $v$ a neighbor of $x$. If $v$ has a neighbor $u \neq x$ such that the component of $T\setminus vu$ containing $u$ has linear MIM-width at least $k$, then we call $v$ a $k$-\emph{neighbor} of $x$.
\end{definition}

\section{Results}

We repeat two structural results from \cite{hogemo2019linear} here:

\begin{lemma}[path layout lemma (\cite{hogemo2019linear}, Lemma 1)]\label{pathlayout}
Let $T$ be a tree. If there exists a path $P = x_1,\ldots,x_p)$ in $T$ such that every connected component of $T\setminus N[P]$ has linear MIM-width $\leq k$, then $lmw(T) \leq k+1$.
\end{lemma}

\begin{lemma}[\cite{hogemo2019linear}, Theorem 1]\label{kneighbors}
Let $T$ be a tree. $lmw(T) \geq k+1$ if and only if there is a node $x\in T$ that has at least three $k$-neighbors.
\end{lemma}

Lemma \ref{kneighbors} will be our main tool for recursively generating trees with a certain linear MIM-width. The next lemma is a generalization of its backwards direction:

\begin{lemma} [due to Vågset~\cite{vagsetmimwidth}]\label{threeparts}
Let $G$ be a graph, and let $C_1$, $C_2$ and $C_3$ be connected induced subgraphs of $G$ with pairwise distance at least two. Let $k$ be the minimum linear MIM-width of these three subgraphs. If, for each pair of subgraphs $C_i,C_j$ there exists a path $P_{i,j}$ that runs from $C_i$ to $C_j$ without intersecting the closed neighborhood of the third subgraph, then the linear MIM-width of $G$ is strictly greater than $k$.
\end{lemma}
\begin{proof}
To prove this lemma, we assume towards a contradiction that there exists a linear layout $\sigma$ of $G$ with MIM-width $k$.

By definition of linear MIM-width, $\sigma$ must contain three cuts $V^\sigma_{i_1}$, $V^\sigma_{i_2}$ and $V^\sigma_{i_3}$, such that $G[V^\sigma_{i_1},\overline{V^\sigma_{i_1}}]$ contains an induced matching $M_1$ with $k$ edges from
$C_1$; likewise there exists an $M_2$ of size $k$ in
$C_2$, and $M_3$ in $C_3$. Since the subgraphs have distance at least 2, any edge from (say) $C_2$ can increase the size of the matchings $M_1$ or $M_3$. Assuming that $(G,\sigma)$ has MIM-width $k$, it must thus be the case that either, for every vertex $x\in C_2$, $\sigma(x) \leq i_1$, or that, for every vertex $x\in C_2$, $\sigma(x) > i_1$. These facts are obviously also true of every other pair of subgraphs.

Now we assume w.l.o.g. that $i_1 < i_2 < i_3$. This implies that some vertex in $C_1$ lies to the left of the cut $V^\sigma_{i_2}$ and that some vertex in $C_3$ lies to the right of the cut. From the previous fact, we can directly infer that \emph{all} vertices of $C_1$ (resp. $C_3$) lie to the left (resp. right) of $V^\sigma_{i_2}$. This means that some edge $e$ on the path $P_{1,3}$ must cross the cut. But the path has also distance at least 2 to $C_2$, thus $e$ can be taken into $M_2$; this implies that $mw(G,\sigma) \geq k+1$. By contradiction, the above lemma is true.\qed
\end{proof}

The last lemma we will use is a special case of Theorem 5 of \cite{jaffke2019mim}, stated in terms of linear layouts; in the original paper this result is given in terms of branch decompositions.

\begin{lemma}
Given a graph $G$ and a layout $\sigma$ such that $G$ has MIM-width $k$ under $\sigma$, then for any power $G^m$ of $G$, $G^m$ has MIM-width at most
$2k$ under $\sigma$. Therefore, $lmw(G^m) \leq 2\cdot
lmw(G)$.
\end{lemma}
\begin{proof}
This follows directly from \cite{jaffke2019mim}, Theorem 5, that states that the property holds for any cut of the graph. Therefore it must hold also for linear layouts.\qed
\end{proof}

\begin{theorem}
For any tree $T$ with $lmw(T) = k$, then $k \leq lmw(T^2) \leq 2k$. These bounds are tight for any $k \geq 0$.
\end{theorem}
\begin{proof} To prove the first inequality, we use induction to prove
that for every tree $T$, if $lmw(T) \geq k$, then $lmw(T^2) \geq k$.

For the base case, it is trivial to see that if $T$ has linear MIM-width at least 1 (that is, if $T$ contains at least one edge), then $T^2$ must also have linear MIM-width at least 1.

For the inductive step, we fix a $k \geq 1$, and assume that for every tree $T$ with $lmw(T) \geq k$, $lmw(T^2) \geq k$. We show that for any tree $T'$ with $lmw(T') \geq k+1$, $lmw(T'^2) \geq k+1$ as follows:

Let $T$ be a tree with $lmw(T) \geq k+1$. From Lemma 2, we know that there must exist a node $x\in T$ that has at least 3 $k$-neighbors in $T$, i.e. $x$ has neighbors $v_1,v_2,v_3$ such that there are three subtrees $S_1,S_2,S_3 \subseteq T\setminus N[x]$ adjacent to $v_1$, $v_2$ and $v_3$ respectively, with $lmw(S_1),lmw(S_2),lmw(S_3) \geq k$.

By the inductive assumption, $lmw(S^2_1),lmw(S^2_2),lmw(S^2_3) \geq k$. $S^2_1$, $S^2_2$ and $S^2_3$ are three connected induced subgraphs of $T^2$, all of distance at least two from each other. Furthermore, between each two of the subtrees -- say $S^2_1$ and $S^2_2$ -- there exists a path in $T^2$ that does not intersect the closed neighborhood of the third subtree, in this case $N[S^2_3]$. By Lemma 3, $T^2$ must have linear MIM-width at least $k+1$.\\

The second inequality follows directly from Lemma 4.\\

Next, we show the downward tightness of the bound; that is, that there exists an infinite family of trees $$\mathcal{L} = (L(0),L(1),\ldots)$$ where $lmw(L(k)^2) = lmw(L(k)) = k$ for every $k\geq 0$. For ease of notation, we will define each tree in $\mathcal{L}$ as a rooted tree.

$L(0)$ is defined to be the singleton $K_1$. For every $k \geq 0$, $L(k+1)$ has a root $u$ with three children, $v_1,v_2,v_3$. Each $v_i$ in turn has one child that is the root of a copy of $L(k)$, which we call $S_i$. This recursive structure enables us to show that between any two trees $L(k)$ and $L(k+1)$, the linear MIM-width must increase with at least 1, due to Lemma 2. On the other hand, between any $L(k)^2$ and $L(k+1)^2$, the linear MIM-width must increase with \emph{at most} 1. This is shown by constructing a layout of $L(k + 1)^2$ of MIM-width $k+1$; this layout is identical to the one given in the proof of the Path Layout Lemma (see \cite{hogemo2019linear} for detatils).

We prove the following claim by induction: For any tree $L(k)\in \mathcal{L}$, $lmw(L(k)^2) = lmw(L(k)) = k$.

The base case is the trivial observation that $lmw(K^2_1) = lmw(K_1) = 0$.

For the inductive step, we assume that $lmw(L(k)^2) = lmw(L(k)) = k$ for some $k \geq 0$, and show that $lmw(L(k+1)^2) = lmw(L(k+1)) = k+1$. From the struc-ture of $L(k+1)$ and the induction hypothesis, it is evident that the root $u$ has three k-neighbors, namely all its children.

By Lemma \ref{kneighbors}, we can conclude that $lmw(L(k+1)) \geq k+1$. Regarding $L(k+1)^2$, by the induction hypothesis there exists a layout of $L(k)^2$ that has MIM-width exactly $k$. We thus have optimal layouts $\sigma_{S^2_1},\sigma_{S^2_2},\sigma_{S^2_3}$ available. We construct a layout $\sigma_{L(k+1)}$ that has MIM-width exactly $k+1$ as follows:
$$\sigma_{L(k+1)^2} = (u) \oplus \sigma_{S^2_1} \oplus (v_1) \oplus \sigma_{S^2_2} \oplus (v_2) \oplus \sigma_{S^2_3} \oplus (v_3)$$
where $\oplus$ signifies concatenation.

For any cut $(V^\sigma_i,\overline{V^\sigma_i})$, a maximum matching contains at most $k$ edges from within some $S^2_a$. How many edges from outside $S^2_a$, i.e. in the graph
$$G'_i := L(k+1)^2[V^\sigma_i,\overline{V^\sigma_i}] - E(S^2_i)$$
can be taken into a matching? We see that every vertex in $V(G'_i)\cap V^\sigma_i$ is in $\{u\}\cup\{v_1,\ldots,v_{a-1}\}\cup(S_a\cap V^\sigma_i)$, or has no neighbors. Every vertex in $S_a\cap V^\sigma_i$ only has at most $v_a$ as neighbor, $v_1,\ldots,v_{a-1}$ have $v_a,\ldots,v_3$ as neighbors, and finally $u$ has all of these and possibly also one vertex in each of $S_a,\ldots,s_3$ as neighbors. This implies that $G'_i$ is a bipartite chain and has MIM 1. Thus, no induced matching in $L(k+1)^2[V^\sigma_i,\overline{V^\sigma_i}]$ can have size more than $k+1$.

Now we have that $lmw(L(k+1)^2) \leq k+1 \leq lmw(L(k + 1))$. But, as we have proven above, $lmw(L(k+1)^2) \geq lmw(L(k+1))$. Thus, $lmw(L(k+1)^2) = lmw(L(k+1)) = k+1$, and every tree in $\mathcal{L}$ has the same linear MIM-width as its square.\\

Finally, we show the upward tightness of the bound; that is, that that there exists an infinite family of trees $\mathcal{H} = (H(0),H(1),\ldots)$ where $lmw(H(k)^2) = 2\cdot lmw(H(k)) = 2k$ for every $k\geq 0$. We will also define each tree in $\mathcal{H}$ as a rooted tree.

$H(0)$ is again the singleton $K_1 = u_0$. For every $k \geq 0$, $H(k+1)$ has a root $u_k$ with three children, $v_1,v_2,v_3$. Each $v_i$ in turn has three children, each of which is the root of a copy of $H(k)$. We call these trees $S_{i,1}$, $S_{i,2}$ and $S_{i,3}$. This recursive structure enables us to show that between any two trees $H(k)$ and $H(k+1)$, the linear MIM-width must increase with at most 1, due to Lemma \ref{pathlayout}: Taking the path to be $(u_{k+1})$, we see that all the subtrees that remain after removing the neighborhood of $u_{k+1}$ are the $S_{i,a}$ for $1\leq i,a\leq 3$. Since, by assumption, $S_{i,a}$ has linear MIM-width $k$, $lmw(H(k+1)) \leq k+1$. And in fact, it is exactly $k+1$ since $H(k)$ is a supertree of $L(k)$ for every $k$.

On the other hand, between any $H(k)^2$ and $H(k+1)^2$, the linear MIM-width must increase with at least 2. This is shown by applying Lemma \ref{threeparts} twice. To this end, we must show that the linear MIM-width of $H(k)^2$ does not decrease when removing its root $u_k$; this trick is to ensure that the subgraphs we consider are situated far enough apart that Lemma \ref{threeparts} is applicable.

We will now prove the following claim by induction: Given that $uk$ is the root of some tree $H(k)$, $lmw(H(k)^2) = lmw(H(k)^2\setminus\{u_k\}) = 2\cdot lmw(H(k)) = 2k$. Note that the graph $H(k)^2\setminus\{u_k\}$ is still a connected graph.

For the base case, it is clear that $lmw(k^2_1) =
lmw(u_0) = lmw(\emptyset) = 2\cdot lmw(K_1) = 0$.

For the induction step, we assume that $lmw(H(k)^2) = lmw(H(k)^2\setminus\{u_k\}) = 2\cdot lmw(H(k)) = 2k$ for some tree $H(k)$ with root $u_k$, and show that $lmw(H(k)^2) = lmw(H(k+1)^2\setminus\{u_{k+1}\}) = 2\cdot lmw(H(k+1)^2) = 2(k+1)$. We know from the induction hypothesis that for every $(S_{i,a})^2$, the graph $S'_{i,a} = (S_{i,a})2\setminus\{u_k\}$ is a connected subgraph of $H(k)^2$ with linear MIM-width $2k$. Thus, the three graphs $S'_{i,1}$, $S'_{i,2}$ and $S'_{i,3}$ are three connected subgraphs with pairwise distance two in the graph $T'_i= (H(k+1)[v_i])^2\setminus\{v_i\}$. Furthermore, for each pair of sub-graphs $S'_{i,a}$ and $S'_{i,b}$, there is a path between them that does not intersect the neighborhood of $S_{i,c}$ (the red path in the illustration below). By Lemma \ref{threeparts}, the linear MIM-width of every $T'_i$ is at least $2k+1$. (In the case $k = 0$, the notion of paths between $S'_{i,a}$ and $S'_{i,b}$, which are empty sets, does not really make sense. In this case, just note that $T'_i$ contains at least one edge and thus must have linear MIM-width at least 1.)

We use the same argument one more time: $T_1$, $T'_2$ and $T'_3$ are three connected subgraphs with pairwise distance two in the graph $H' := H(k+1)^2\setminus\{u_{k+1}\}$. For each pair $T'_a$ and $T'_b$, there is a path between them that does not intersect the neighborhood of the third subgraph $T'_c$ (the blue path in the illustration below). Thus, $lmw(H') \geq 2k+2$. Since $H'$ is an induced subgraph of $H(k+1)^2$, we have the situation that
$$lmw(H(k+1)^2) \geq lmw(H') \geq 2(k+1) = 2\cdot lmw(H(k+1))$$
But, as we have noted above, $lmw(H(k+1)^2) \leq 2\cdot lmw(H(k+1))$. Thus,
$$lmw(H(k+1)^2) = lmw(H(k+1)^2\setminus\{u_{k+1}\}) = 2\cdot lmw(H(k+1)) = 2(k+1)$$
and every tree in $\mathcal{H}$ has half the linear MIM-width of its square.\qed
\end{proof}

\begin{figure}[bt]
    \centering
    \caption{A subgraph of the graph $H(k+1)^2$. Dashed lines
    indicate power edges.}
    \includegraphics[width=\linewidth]{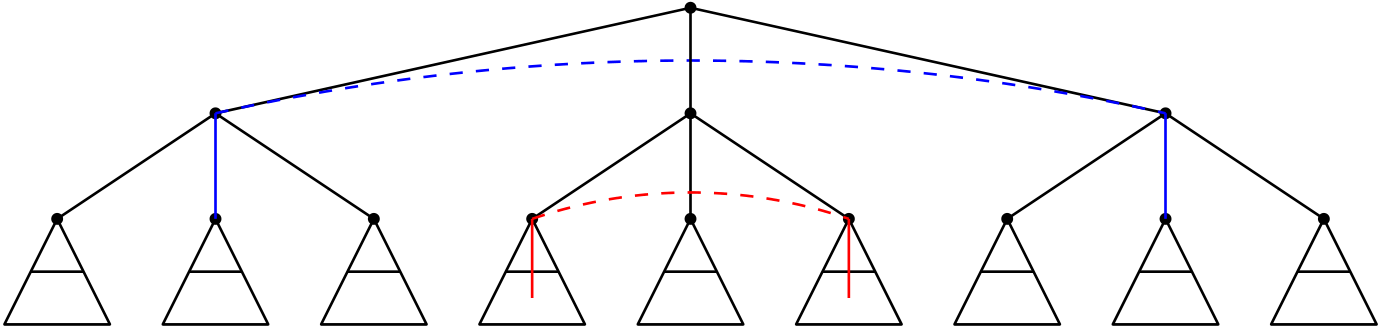}
\end{figure}

\section{Conclusion}

We have shown that there is little connection between the linear MIM-width of a tree and that of its square, except the fact that it cannot decrease. This fact is interesting, as it implies that taking the square of a tree does not makes its vertices more well-connected than in the original tree. As we know, there must, for any tree $T$ with $lmw(T)\geq 2$, be an exponent $k$ such that $lmw(T^k)<lmw(T)$, since for any finite graph $G$, $G^{diam(G)}$ is a complete graph and thus has linear MIM-width 1. How high this exponent must be to decrease the linear MIM-width (or bring it down to 1), and whether there exists a poly-time algorithm to evaluate the linear MIM-width of powers of trees, must still be left in the open.

\section{Acknowledgements}

The author would like to thank O-joung Kwon for the initial discussion of this topic.

\bibliographystyle{splncs04}
\bibliography{refs}

\end{document}